\documentclass[12pt]{article}
\usepackage{}
\usepackage{amsfonts}
\usepackage{amssymb}
\usepackage{mathrsfs}
\usepackage{amsmath,amssymb,amsthm,amscd}
\usepackage{epsf,amsfonts,hyperref}
\usepackage{color}
\usepackage{graphicx}

\input epsf.sty
\newcommand{\be}{\begin{equation}}
\newcommand{\ee}{\end{equation}}
\newcommand{\ba}{\begin{aligned}}
\newcommand{\ea}{\end{aligned}}
\newcommand{\bea}{\begin{eqnarray}}
\newcommand{\eea}{\end{eqnarray}}
\newcommand{\Res}{\mathop{\,\rm Res\,}}

\newtheorem{definition}{Definition}[section]

\newtheorem{proposition}{Proposition}[section]

\newtheorem{corollary}{Corollary}[section]

\begin{document}

\title{{\bf Hopf Algebraic Structure for Tagged Graphs and Topological Recursion}}
\author{\; Xiang-Mao Ding$^{1,2}$\footnote{\small
Email: xmding@amss.ac.cn},\;  Yuping Li$^1$\footnote{\small
Email: liyuping@amss.ac.cn},\; Lingxian Meng$^{3,1}$\footnote{\small Email: menglingxian@amss.ac.cn}
\\{\it ~$^1$ Institute of Applied Mathematics, Academy of
}\\ {\it  Mathematics and Systems Science; Chinese Academy of Sciences,}
\\{\it ~$^2$ Hua Loo-Keng Key Laboratory of Mathematics, Chinese Academy of Sciences}
\\{\it Beijing 100190, People's Republic of China}\\{\it ~$^3$College of Mathematics and Information Science,}\\{\it Zhengzhou University of Light Industry}\\{\it Zhengzhou 450002, People's Republic of China}\date{}}


\maketitle

\begin{abstract}
Using the shuffle structure of the graphs, we introduce a new kind of the Hopf algebraic
structure for tagged graphs with, or without loops. Like a quantum group structure, its
product is non-commutative. With the help of the Hopf algebraic structure, after taking
account symmetry of the tagged graphs, we reconstruct the topological recursion on spectral
curves proposed by B. Eynard and N. Orantin, which includes the one-loop equations of
various matrix integrals as special cases.
\vspace{0.3cm}


\vspace{0.3cm}

\end{abstract}

\section{Introduction }

Correlation functions are basic quantities to be defined in a lot of physical systems. For a matrix model, the zero-dimensional quantum system, the correlation functions are given by various matrix integrals. As observed
independently by G. 't Hooft, E. Brezin and C. Itzykson etc  \cite{Gh, BIPZ},
the integrals contain abundant useful information, which is important either in mathematics, or physical research and applications. The quantities of matrix integrals admit a topological expansion, in certain cases, they enumerate discrete colored surfaces for a given topology, while the well known Kontsevich integrals count spin intersection numbers. In appropriate double scaling limits, some integrals involve in a Liouville theory. Hence the study of various matrix integrals is one of the most important research fields in mathematical physics.

Various sorts of working approaches have been set up in the past few decades to deal with the matrix integrals, such as orthogonal polynomials method \cite{FGJ}, loop equations \cite{JLCY} and topological string method \cite{DV1, DV2}. Among them, the loop equation method is one of the most effective ways at present in some sense. The loop equations are known as well as the Schwinger-Dyson equations for certain cases. The original form of the loop equations are emerged in the study of the Yang-Mills gauge theory \cite{Ym1,Ym2}. Then similar equations are obtained for cases of matrix model. In fact, In the matrix model cases, the loop equations are the functional identities for matrix integrals. The loop equations for the hermitian matrix integrals are given in  \cite{Sr}. For a long period of time, the loop equation is a key step in the study of matrix integrals in the double scaling limit. In order to solve the problems which are away from the double scaling limit, various kinds algorithmic methods for the calculating the higher genus contributions were suggested in the following decades. In practice, if the genus zero equation was solved, then the higher-genus equations were iteratively solved genus by genus by using the results of previous steps. Even though these algorithms were ingenious, nevertheless, only several lower order loop equations could be solved. It is very difficult to compute the results concretely for higher-genus cases.

A new type of loop equations was proposed by B. Eynard for one-hermitian matrix integrals \cite{Be}. After then, it was generalized to 2-matrix model case in the refs. \cite{EO2,CEO}. In contract to the original loop equation, the new version considers correlation functions on an algebraic curve given by the one-loop equation. The new approach combines elements which is depending on the algebraic curve.  It allows us to compute the integrals by parts on the algebraic curve recursively, therefore the solutions have plentiful mathematical structures. The method can be applied to other problems which are irrelevant to matrix integrals as well. This implies that the structures of the solutions do not depend on the matrix integrals. In the ref. \cite{EO} B. Eynard and N. Orantin constructed the so called topological recursion structure for spectral algebraic curves. The topological recursion is a recursive defined identities, which associates a double indexed family differential forms ${\omega}_{g,n}$, with the two integers $g$ and $n$ being non-negative. The ${\omega}_{g,n}$ are called the invariants of the spectral curve. Those invariants enjoy lots of fascinating properties, one could see ref. \cite{EO} for more details. If an algebraic curve is identical to the one-loop equation of certain matrix integrals, then these new constructed correlation functions coincide with the those given by matrix integral.

The expression for perturbative expansion in a quantum field theory(QFT) has a similar structure as the topological recursion. The perturbative expansion in QFT can be represented by the Feynman diagrams. There are two kinds of Feynman diagrams: tree diagrams and loop ones. In the Feynman diagrams approach for a quantum field theory, the propagators are represented by lines, and in addition,  interactions are represented by vertices. It was found that there was a Hopf algebraic structure hidden in the relations of divergences in the expansions of amplitudes \cite{Dk1, Dk2}. The Hopf algebra was graded according to their sub-divergences. Moreover, regulation and normalization of the divergences were compatible with the Hopf algebraic structure as well \cite{CK1, CK2}. In terms of algebraic language, these procedures were the Birkhoff decomposition, either the Schwinger-Dyson equation or the renormalization group equation can be restated in this terminology. With the help of the Hochschild cohomology and the Schwinger-Dyson equation, the relationships of the contributions of a single graph to the full correlation functions were established. In addition with the Hopf algebra, the computations of amplitudes can be simplified to a certain extent, even non-perturbative results could be obtained in some cases.

Similar to the perturbative expansion in quantum field theory, the topological recursion on an algebraic curve has a diagrammatic representation. The diagrammatic representation for the topological recursion is the Feynman diagram alike, although it is not a Feynman diagram.
The main disparity between them is that, the diagrammatic representation of topological
recursion is arrowed.

It is straightforward to ask the question, whether there is a Hopf algebraic structure on the topological recursion.  Indeed, there is a Hopf algebraic structure on the tree diagrammatic representation with the help of the Loday-Ronco Hopf algebra  \cite{Est}. However, that structure cannot be extended to graphs with loops.

In this manuscript we present a Hopf algebraic structure on the topological recursion.
Distinct from the Loday-Ronco Hopf algebra, we redefine the product and the coproduct on planar binary trees with tagged leaves, and a new Hopf algebra is formed, and more remarkably, the new Hopf algebraic structure can be generalized to graphs with loops. It
seems that the numbers of the tagged diagrams are not the same as the diagrammatic representations of the topological recursion. However, If we consider the symmetries
of the tagged diagrams, these graphs are just the diagrammatic representations of topological recursion. With the help of the new Hopf algebraic structure on the tagged graphs with loops,
and in addition with the weighted map defined by B. Eynard and N. Orantin as well  \cite{EO}, we can obtain a Hopf algebraic structure on the topological recursion from its diagrammatic representation.

The manuscript is arranged as follows. In section 2, we introduce basic notations and briefly review the main results of Eynard-Orantin's topological recursion in ref. \cite{EO}. We briefly present its diagrammatic representation. In the next section, we recall the Loday-Ronco Hopf algebra defined on planar binary trees. In section 4, we prove the equivalence between the diagrammatic representation of the topological recursion and the graphs given in  \cite{Est}. The main part of the manuscript is in the following section. We give a new Hopf algebraic structure on planar binary trees. It is straightforward to extend our structures to the graphs with loops. Using the weighted map, we really get a Hopf algebraic structure on the topological recursion. In section 6, conclusions and further discussions are followed.

\section{Topological recursion and diagrammatic representations}

The primary version of loop equations derived by B. Eynard gives an effective recursive way to solve formal one-hermitian matrix integrals  \cite{Be}, the solutions depend on the one-loop equation rather than formal matrix integrals. Then B. Eynard and N. Orantin generalized the results for one-loop equation to cases for algebraic curves, on which they constructed a sequence of differential forms, or geometric invariants named as correlation functions, furthermore a sequence of complex numbers regarded as partition functions by the topological recursion, respectively. Here, we review the basic notations and main results of  \cite{EO}, and for the limit of space, we only focus on the correlation functions. The
argument can be applied to the cases of partition functions.

Let us consider a compact Riemann surface $\Sigma$ with genus $g$, $x$ and $y$ are meromorphic functions on $\Sigma$. For a local variable $p$, $dx(p)=0$ gives the branch points $\{a_i|i=1,...,n\}$. Assume that all the branch points $\{a_i|i=1,...,n\}$ are simple.

On the Riemann surface $\Sigma$, basic ingredients in the correlation functions can be defined. Consider any point $p\in\Sigma$ and points $q,\ \overline q$ in the vicinity of a
fixed branch point, such that $x(q)=x(\bar q)$. The are two types of basic ingredients for the topological recursion:  one is ``vertex" , and another is the Bergmann kernel \cite{EO}.
The ``vertex" is defined as:
\begin{equation}
\omega(q)=(y(q)-y(\overline q))dx(q).
\end{equation}

\noindent While the Bergmann kernel $B(p,q)$ is a symmetric 2-form on $\Sigma\times\Sigma$. There is exactly one double pole at $p=q$ with no residue. For a canonical basis of cycles $(A_I,B^I)$, such that
$\int_{A_I}B=0$.

Given a set of points $\{p_{1},p_{2},\cdots,p_{k}\}$ of the curve, if $J=\{i_1,i_2,\cdots,i_j\}$ is any a subset of $K=\{1, 2,\cdots, k\}$, denote $P_{J}=\{p_{i_1},p_{i_2},\cdots,p_{i_j}\}$. Then the $(k+1)$-point correlation function to order $g$ has the
following combinational structure:
\begin{equation}\label{defWk}
\begin{aligned}
& W_{k+1}^{(g)}(p,P_{K}) \\
&=\Res_{q\to a} K(q,p)\,\Bigl(\sum_{m=0}^g \sum_{J\subset
K} W_{|J|+1}^{(m)}(q,P_J)W_{k-|J|+1}^{(g-m)}(\overline{q},P_{K/J}) +
W_{k+2}^{(g-1)}(q,\overline{q},P_K) \Bigr)
\end{aligned}
\end{equation}
\\ where
\begin{equation}
K(q,p)={{1\over 2}\int_q^{\bar q}B(\xi,p)\over \omega(q)}
\end{equation}
the integration path in $K(q,p)$ is chosen in the neighborhood of a branch point.
The recursive basis are
\begin{equation}
W_k^{(g)}=0 \quad {\rm
for}\,\, g<0
\end{equation}
\begin{equation}
W_1^{(0)}(p) = 0
\end{equation}
\begin{equation}
W_2^{(0)}(p_1,p_2)
= B(p_1,p_2)
\end{equation}

\noindent It can be proved that the function $W_{k+1}^{(g)}(p,p_1,\dots,p_k)$ is symmetric  with respect to its variables \cite{EO}.

The correlation functions, as well as the partition functions defined on an algebraic curve have graphic representations. B. Eynard and N. Orantin introduced a set of graphs
$\mathscr{G}_{k}^{g}(p,p_1,\dots,p_k)$(In \cite{EO}, it was denoted as $\mathscr{G}_{k+1}^{g}(p,p_1,\dots,p_k)$). Pay your attention that the formulation is similar to the Feynman diagram expression in quantum field theories, however it is not the Feynman diagram.
These graphs have two indices $k,\ g$, such that $k\geq 0,g\geq 0$ and $k+2g\geq 3$.
For a graphic representations, $\mathscr{G}_{k}^{g}(p,p_1,\dots,p_k)$ is a set of connected trivalent graphs which are subjected to the following three conditions:

(I),  there are $2g+k-1$ vertices which are the trivalent vertices, there are $k$ leaves which are the 1-valent vertices labelled with $p_1,\dots, p_k$, and there is one root which is the 1-valent vertex labelled by $p$.

(II), there are $3g+2k-1$ edges, among them $2g+k-1$ edges are arrowed, and $k+g$ are non-arrowed. The arrowed edges form a skeleton tree with the root $p$. The arrows from root to leaves induce a partial order $\mathscr{R}$ of the vertices. Among the non-arrowed
edges,  there is $k$ ones from a vertex to a leaf which are named as outlines, and the
remaining $g$ edges connect the pairs of vertices with the $\mathscr{R}$ relations.

(III), for two edges, if one is arrowed,  another one is non-arrowed, and they have a common vertex, furthermore, the non-arrowed edge links this vertex with its descendant. Then we set the arrowed edge to be the left child, and the non-arrowed one be the right child with labeling by a black dot in the graph.

We denote the set of all graphs of the correlation functions by
$$\mathscr{G}\equiv \cup_{k+2g\geq 3,k\geq 0,g\geq
0}\mathscr{G}_{k}^{g}(p,p_1,\dots,p_k).$$
\noindent The graphs in
$\mathscr{G}$ are relevant to the correlation functions by the weighted map given by B. Eynard and N. Orantin \cite{EO}. The ingredients in the correlation functions are
expressed as basic elements of the graphs, the weighted map $\phi$ is defined as:
$\phi (\begin{array}{r}
{\epsfxsize 2cm\epsffile{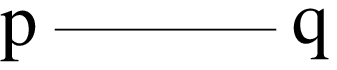}}
\end{array})=B(p,q)$;
$\phi (\begin{array}{r} {\epsfxsize 3cm\epsffile{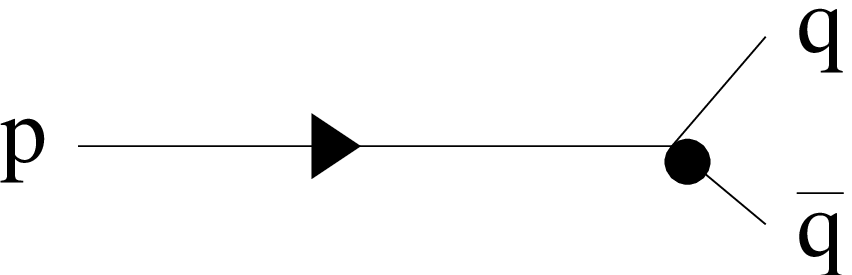}}
\end{array})=K(p,q)$;
and take a residue $q\longrightarrow \{a_i\}$ at each vertex $q$, and sum over all the branch points along the arrow from leaves to the root. Under the weighted $\phi$ they proved the relationship of the graphs and the correlations \cite{EO}:
\begin{equation}\label{rules}
W_{k+1}^{(g)}(p,p_1,\dots,p_k)=\phi(\sum_{G\in
\mathscr{G}_{k}^{g}(p,p_1,\dots,p_k)} G)=\sum_{G\in
\mathscr{G}_{k}^{g}(p,p_1,\dots,p_k)}\phi(G)
\end{equation}

For a given vertex in $\mathscr{G}$, if we exchange the left children with the right ones, the images under $\phi$ are intact. This will be encoded into symmetry factors. By equation (\ref{rules}), the
elements of $W_{k+1}^{(g)}(p,p_1,\dots,p_k)$ are one-to-one
corresponding to the graphs in
$\mathscr{G}_{k}^{g}(p,p_1,\dots,p_k)$. Hence we will not distinguish
them from now on. By the definition of $\mathscr{G}$, such diagrammatic representation $\mathscr{G}_{0}^{2}$ of $W_{1}^{(2)}$
contains the graphs given in Figure 1.
\begin{figure}\label{g12new}
\centering $
\begin{array}{r}
{\epsfxsize 10cm\epsffile{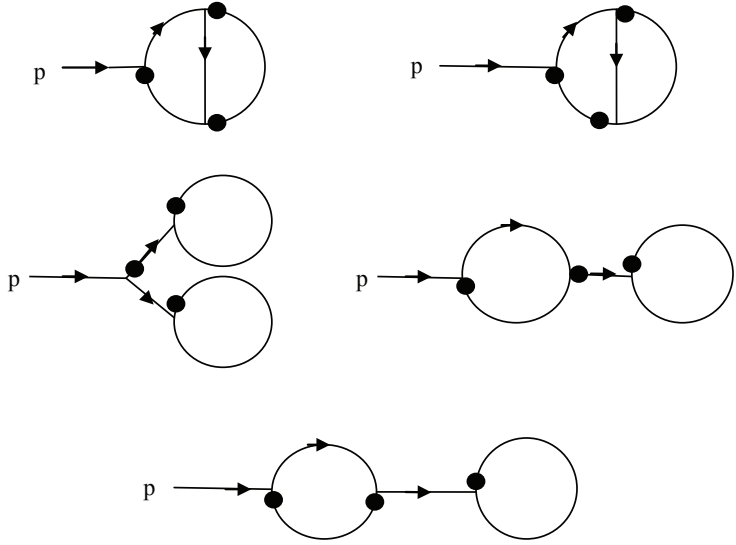}}
\end{array} $
\caption{Graphs in $\mathscr{G}_{0}^{2}$ and the $\bullet $ denotes
the right child.}
\end{figure}

\section{The Loday-Ronco Hopf algebra of planar binary trees}

In this section, we review the Loday-Ronco Hopf algebra of planar binary trees. Details and proofs can be found in  \cite{LR}.

\subsection{The Shuffle algebra over the symmetric groups }

Let $S_n$ be a set of permutations of order $n$. It is a symmetric group acting on $\{1,\cdots, n\}$. The trivial group $S_0$ is nothing but the unit element $1$. An element $\rho\in S_n$ is denoted by its image $(\rho(1)\cdots \rho(n))$. For a field $\mathbb{K}$, denoting the group algebra as $\mathbb{K}[S_n]$ which is generated by the $S_n$ over $\mathbb{K}$. Furthermore, we set the disjoint union of $S_n$ for $n\geq 0$ as $S_{\infty}$, hence

$$\mathbb{K}[S_{\infty}]=\oplus_{n\geq 0}\mathbb{K}[S_n].$$

\noindent It is a graded vector space. An $(n,m)-$shuffle is defined as an element $\sigma\in S_{n+m}$ such that

\bea
&\sigma(1)<\sigma(2)<\cdots<\sigma(n)\nonumber\\
&\sigma(n+1)<\sigma(n+2)<\cdots<\sigma(n+m).
\eea

\noindent  The set of all $(n,m)-$shuffles is denoted by $Sh_{n,m}$. For any permutations $\rho\in S_m$ and $\sigma\in S_n$, the product $\rho\times\sigma$ on $S_{\infty}$ is the  permutation in $S_{n+m}$, which is obtained by limit the $\rho$
action on the former $m$ variables, and $\sigma$ action on the remain $n$ variables.

There exists a graded Hopf algebraic structure on $\mathbb{K}[S_{\infty}]$. For $x\in \mathbb{K}[S_n]$, $y\in \mathbb{K}[S_m]$, the product is defined as

\begin{equation}
\label{perprod} x*y= \sum_{\alpha_{n,m}\in Sh_{n,m}}\alpha_{n,m}\circ(x\times y).
\end{equation}

\noindent For any one $\sigma\in \mathbb{K}[S_n]$, the coproduct can be defined as

\begin{equation}
\label{permcop} \Delta(\sigma)=\sum_{i=0}^n\sigma_i\otimes
\sigma'_{n-i}.
\end{equation}

\noindent  where $\sigma_i\in \mathbb{K}[S_i],\sigma'_{n-i}\in
\mathbb{K}[S_{n-i}]$ such that

\begin{equation}\label{codecom}
\sigma=(\sigma_i\times\sigma'_{n-i})\circ
\omega^{-1},
\end{equation}

\noindent where $\omega$ is a $(i,n-i)-$shuffle. For a given element $\sigma\in \mathbb{K}[S_n]$ and a number $i$, it can be
proved that there exists a unique $(i,n-i)-$shuffle $\omega$ and $\sigma_i\in
\mathbb{K}[S_i], \sigma_{n-i}\in \mathbb{K}[S_{n-i}]$, such that the identity
(\ref{codecom}) is held \cite{LR}. By using the definition of product $*$
and coproduct $\Delta$, $\mathbb{K}[S_{\infty}]$ becomes a bialgebra,
moreover, it is not only graded but also connected, therefore $\mathbb{K}[S_{\infty}]$ is a Hopf algebra \cite{MR}.

\subsection{Planar binary trees and Hopf algebra}

A planar binary tree $t$ is a special kind of the binary tree, which is an oriented graph on a plane without loops, all of its internal vertices
are trivalent. Let $|t|$ denote the degree of a planar binary tree $t$, which is the number
of trivalent vertices in $t$. For every planar binary tree, there is a root, which is a particular edge with 1-valent vertex, the rest edges with 1-valent vertex are leaves. Obviously, the tree $t$ has $|t|+1$ leaves. Following the notation in  \cite{LR}, we denote the set of planar binary trees with $n$ vertices by $Y_n$. The disjoint union of $Y_n$ for $n\geq 0$ is denoted by $Y_{\infty}$. Similarly, $\mathbb{K}[Y_{\infty}]=\oplus_{n\geq 0}\mathbb{K}[Y_n]$.

For simplicity, a planar binary tree with levels is mentioned as a planar binary tree $t$, in which every internal vertex is assigned to one of the horizontal lines. The horizontal lines are entangled from top to bottom by the numbers $\{1,\cdots,|t|\}$. Hence, each vertex has a unique number. By this way, the level of the vertex attached to the root is $|t|$. For simplicity, we denoted them as $\tilde Y_n$, for the set of all planar binary trees which levels are $n$.

For a planar binary tree with levels, if we order its vertices from the left to right and from the top to bottom, then it gives a permutation \cite{Est}. The inverse is established in the same way. In other word, $\tilde Y_n$ is one-to-one corresponding to $S_n$ \cite{LR}. One can try out the proof in ref. \cite{LR}.

In order to define the Hopf algebra over $\mathbb{K}[Y_{\infty}]$, a grafting operator has not defined yet. In fact, for any two trees $t_1$ and $t_2$, the grafting $t_1\vee t_2$ fuse into a new tree $t$ with the left branch $t_1$ and the right one $t_2$,Therefore $|t|=|t_1|+|t_2|+1$. If we denote the one-leaf tree without trivalent vertices by $|$. And any picked tree can be expressed by two branches with a grafting operator.

The bijection $\tilde Y_n\longrightarrow S_n$ and the forgetful map(forget the levels of vertices) $\tilde Y_n\twoheadrightarrow Y_n$ induce a linear map $\psi_n:
S_n\twoheadrightarrow Y_n$. Hence it is defined from $\mathbb{K}[S_n]$ to $\mathbb{K}[Y_n]$ as well. The linear dual maps $\psi_n^*: \mathbb{K}[Y_n]\longrightarrow
\mathbb{K}[S_n]$ are inclusion maps for $n\geq 0$. For instance, if $t\in Y_n$, define the set $Z_t=\{\sigma\in S_n|\psi_n(\sigma)=t\}$. One has $\psi_n^*(t)=\sum_{\sigma\in Z_{t}}\sigma$. Their disjoint union is a graded linear map $\psi^*: \mathbb{K}[Y_{\infty}]\longrightarrow \mathbb{K}[S_{\infty}]$.  J. L. Loday and M. O. Ronco proved that the image of
$\psi^*$ is a Hopf subalgebra of $\mathbb{K}[S_{\infty}]$ \cite{LR}. The product $*$ on $\mathbb{K}[Y_{\infty}]$ satisfies the formula \cite{LR}
\be
\label{plaprod} T*T'=T_1\vee(T_2*T')+(T*T'_1)\vee T'_2 ,
\ee

\noindent where $T$ and $T'$ are arbitrary trees in $\mathbb{K}[Y_{\infty}]$ such that $T=T_1\vee
T_2$ and $T'=T'_1\vee T'_2$. The starting equation of the recursive series is

\be
|*T=T=T*|.
\ee
\noindent The coproduct $\Delta$ on $\mathbb{K}[Y_{\infty}]$ surely satisfies a series of recursive formulas \cite{LR}. For a selcted tree $T=T_1\vee T_2\in Y_{n+m+1}$ with $|T_1|=n$ and $|T_2|=m$,

\be
\label{placop}
\Delta(T)=\sum_{j,k}(T_{1,j}*T_{2,k})\otimes(T'_{1,n-j}\vee
T'_{2,m-k})+T\otimes |,
\ee

\noindent where $\Delta(T_1)=\sum_jT_{1,j}\otimes
T'_{1,n-j}$ and $\Delta(T_2)=\sum_kT_{2,k}\otimes T'_{2,m-k}$.

The two operators product $*$ and coproduct $\Delta$ are obtained by restricting the Hopf algebra on $\mathbb{K}[S_{\infty}]$ to
$\mathbb{K}[Y_{\infty}]$. J. L. Loday and M. O. Ronco proved
they are internal in $\mathbb{K}[Y_{\infty}]$, so the image of
$\psi^*$ is a Hopf subalgebra \cite{LR}.

\section{Equivalence between two kinds of graphs}

We have mentioned two kinds of graphs, one is the diagrammatic
representation of the topological recursion, another one is the planar
binary tree. If we connect two consecutive leaves on a planar binary tree,
it will generates a graph with one loop. These graphs with loops satisfy the
topological recursion formulas \cite{Est}, in fact, they are identical
with the diagrammatic representation of the topological recursion.

If we consider the diagrammatic representation of genus $0$ correlation functions on the topological recursion. For a fixed number $n$ of internal vertices, the cardinality of $\mathcal {G}_{n+2}^0$ is given by the Catalan number $C_n={(2n)!\over n!(n+1)!}$. It is equal to the dimension of $\mathbb{K}[Y_n]$. Intuitively, there is a natural one-to-one
correspondence between them \cite{Est}. The $n+2$ points, genus $0$
correlation function $W_{n+2}^0(p,p_1,\cdots,p_{n+1})$ satisfies the following
equation
\be\ba\label{genus0tree} \phi^{-1}(W_{n+2}^0(p, p_1,\cdots, p_{n+1}))=&\sum\limits_{\scriptstyle t_i\in Y_n \atop
\scriptstyle \text{perm. of leaf labels $\{p_1,\cdots,
p_{n+1}\}$}}t_i=\sum\limits_{\scriptstyle
p+q+1=n,\atop\scriptstyle |t_{1}|=p, |t_{2}|=q}t_{1}\vee t_{2} \\
+&\text{\bf perm. $\{p_1,\cdots, p_{n+1}\}$}, \ea\ee where
$\phi^{-1}$ is the inverse of B. Eynard and N. Orantin's weighted map $\phi$.
The proposition can be proved by the induction on $n$.
It is only to reorganize of the facts in  \cite{Est}.

If one denote by $(Y_n)^g$ the set of distinct graphs with $g$ loops, which are obtained by successively action on $(_{i}\leftrightarrow_{i+1})$ for $g$ times on trees in $Y_n$, then,
similar to genus $0$ case, there is the following equation
\be
\ba
\label{highgenustree}
&\phi^{-1}(W_k^g(p,p_1,\cdots,p_{k-1}))=\sum\limits_{\scriptstyle
t^g\in (Y_n)^g\atop \scriptstyle \text{perm. $\{p_1,\cdots,
p_{k-1}\}$}}t^g \\=&[\sum\limits_{i=0}^g\sum\limits_{\scriptstyle
(t_1)^i\in (Y_p)^i, \ (t_2)^{g-i}\in(Y_q)^{g-i}\atop \scriptstyle
p+q+1=n}(t_1)^i\vee(t_2)^{g-i}
+\sum\limits_{i=0}^{g-1}\sum\limits_{\scriptstyle (t_1)^{g-i-1}\in
(Y_p)^{g-i-1}, \ (t_2)^i\in(Y_q)^i \atop \scriptstyle
p+q+1=n}(t_1)^{g-1-i}\smallsmile (t_2)^{i}]\\&+\text{\bf perm.
$\{p_1,\cdots, p_{k-1}\}$},
\ea
\ee

\noindent in which $n=2g+k-2$, $t=t_1\vee t_2$ and
$t_1\smallsmile t_2$ means the identification between the rightmost leaf in $t_1$ to the leftmost leaf in $t_2$. Hence, $(Y_n)^g=\mathscr{G}_{k-1}^g$.

In order to prove equation (\ref{highgenustree}) is held, we start to set
$g=1$ and make the induction on the number $n$, then make an induction on genus $g$ case. The procedure is only to rearrange of the facts in  \cite{Est}. Equations (\ref{genus0tree}) and (\ref{highgenustree})
imply that the planar trees with their contractions are identified with the diagrammatic representation of the topological recursion.

In order to obtain graphs with loops from planar binary trees, J. N. Esteves defined a contraction \cite{Est}. Given any planar binary tree $t\in Y_n$, label its leaves from left to right by $p_1,\cdots,p_{n+1}$ in turns. The action $(_{i}\leftrightarrow_{i+1})$ on $t$ means attaching an edge to the consecutive leaves $p_i$ and $p_{i+1}$, then relabeling the remaining leaves by $p_1,\cdots, p_{n-1}$. It should be noticed that the only action is the contraction for two the
nearest neighbor leaves.

Based on this equivalence, the Hopf algebra on planar binary trees  can gives a Hopf algebraic structure on the diagrammatic representations of correlation functions with genus $0$. Unfortunately, the structure for the genus $0$ functions cannot be generalized to higher genus cases. In next section, a new Hopf algebra on planar binary trees will be given, while this accompany algebraic structure can be naturally generalized to graphs with loops.

\section{The Hopf algebra of topological recursion}

It is known that, the Loday-Ronco Hopf algebra on planar binary trees cannot
be extended to graphs with loops generated by contractions, in the main, the operator
as of the Loday-Ronco Hopf algebra keep the number of vertices, meanwhile they
change the number of leaves, and the product and coproduct
are not well defined on the graphs with no leaves.

In this section, we give a new form of Hopf algebra on graph with tagged leaves, it can be well defined on graphs with loops after leaves contractions as well. In this way, the planar binary trees is a special case of the graph with tagged leaves without any contractions.
Based on the equivalence of the diagrammatic representation of the topological recursion
and the graph with tagged leaves, as well as using the weight map, in fact, the new Hopf
algebra gives a Hopf algebraic structure on the topological recursion.

In the following, we will give a bracket representation of planar binary trees with tagged leaves in subsection \ref{section5.1}. In the bracket representation, a new Hopf algebra is defined in the
next subsection. Finally, we generalize it to graphs with tagged leaves and loops, and derive the relationship between coproduct and the topological recursion.

\subsection{Bracket representations of planar binary trees}\label{section5.1}

At first, In this subsection, we consider the case of planar binary tree, a special kind of graph without
loop. Let $t$ be a planar binary tree with tagged leaves in $Y_n$. Its leaves are labelled by $1,\cdots, n+1$(or other distinguishable labels) from left to right. We can denote $t$ by a string of numbers with apex angles. We list the numbers $1\cdots n+1$ in succession, then add an apex angle for two adjacent numbers if their corresponding leaves are on the same branch. If regarding an apex angle as a leaf, we keep on the way to add apex angles till that all the numbers are in one apex angle. For example in Figure \ref{fig3} the tree
is $\langle \langle 12\rangle \langle
3\langle 45\rangle \rangle \rangle $.
\begin{figure}
\centering $
\begin{array}{r}
{\epsfxsize 4cm\epsffile{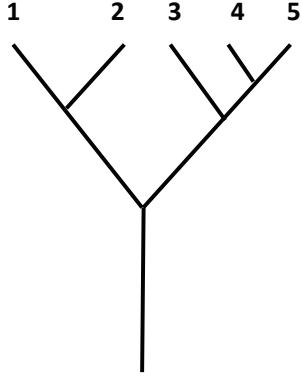}}
\end{array} $
\caption{\label{fig3}planar binary tree $\langle \langle 12\rangle \langle
3\langle 45\rangle \rangle \rangle $ }
\end{figure}

Obviously, any planar binary tree has a unique expression with apex angles. We call these expressions as the bracket representations of planar binary trees. Actually, it can be regarded as a dual representation of the permutation group representation. In fact, from the permutation of a tree $t$, we can get the expression with apex angles directly. If we denote $t=(i_1\cdots i_k)$, its expression with apex angles can be achieved in the following procedure. We list $p_1\cdots p_{k+1}$ in succession, with corresponding to their leaf labels. Then any interval between two labels in $p_1\cdots p_{k+1}$ corresponds a certain number in $(i_1\cdots i_k)$ from left to right. Then we add $k$ apex angles in an arrangement in accordance with the order of the natural numbers in $(i_1\cdots i_k)$. An apex angle which has added is viewed as a new leaf while adding the next pair apex angles.

For an example, assume that $t=(231)$. It has four leaves denoted by $p_1,\ p_2,\ p_3,\ p_4$ from left to right. The number $1$ in $(231)$ corresponds to the interval between $p_3$ and $p_4$, then we add the first apex angle on $p_3p_4$, i.e. $p_1p_2\langle p_3p_4\rangle$. The number $2$ in $(231)$ corresponds to the interval between $p_1$ and $p_2$, hence we add the second apex angle to be $\langle p_1p_2\rangle \langle p_3p_4\rangle$.  At last, the number $3$ is matched with the interval between $p_2$ and $p_3$, i.e. $\langle p_1p_2\rangle$ and $\langle p_3p_4\rangle$. We add the third apex angle to become $\langle \langle p_1p_2\rangle \langle p_3p_4\rangle \rangle $.

Permutations are one-to-one correspondence to planar binary trees with levels, expressions with apex angles are bijective with planar binary trees. Hence, one permutation has a unique expression with apex angles. But one expression with apex angles
can correspond to distinct permutations in general. An example, permutations $(231)$ and $(132)$ have the same expression $\langle \langle 12\rangle \langle 34\rangle \rangle$.

Let $X_n(n\geq 2)$ be a set of digital expressions with apex
angles matching to the following conditions: in any pair of apex angles, there are two elements, a number and a pair of apex angle in a whole viewed as an element; and the whole expression is in a pair of apex angles. Particularly, $X_0=\{\emptyset\}$; $X_1=\{\langle1 \rangle \}$; $X_2=\{\langle 12\rangle \}$. Among them, $\langle1
\rangle$ represents a planar binary tree without trivalent vertices
$|$, $\langle12 \rangle$ stands a planar binary tree with two
leaves, and one trivalent vertex. Graphs with different symbol system have the relationship $X_{n+1}=Y_n=\mathscr{G}_n$.

If we indicate $\mathbb{K}[X_n]$ the vector space generated by $X_n$ over field $\mathbb{K}$. The disjoint union of the sets $X_n$ for $n\geq 0$ is stated
as $X_{\infty}$. Hence, \be \mathbb{K}[X_{\infty}]=\oplus_{n\geq 0}\mathbb{K}[X_n]. \ee

\noindent It should be noticed that $\mathbb{K}[X_{\infty}]$ is not identical to $\mathbb{K}[Y_{\infty}]$, although both they are generated through planar binary trees. In $\mathbb{K}[X_{\infty}]$, $\emptyset$ is included as a generator. But $\mathbb{K}[Y_{\infty}]$ do not include $\emptyset$, this means that
\be
\mathbb{K}[X_{\infty}]=\mathbb{K}[Y_{\infty}]\oplus \mathbb{K}[\emptyset].
\ee

\noindent In the next subsection, we will prove that there is a Hopf algebraic structure on $\mathbb{K}[X_{\infty}]$, which is different from the known Loday-Ronco structure on $\mathbb{K}[Y_{\infty}]$.

\subsection{Hopf algebraic structure for diagrams without loops}

Operators (\ref{perprod}) and (\ref{permcop}), or the expressions in  (\ref{plaprod}) and (\ref{placop}), they surely keep the numbers of trivalent
vertices on planar binary trees. However, the story for leaf is varied.
The trees in the right sides of equations (\ref{perprod})
or (\ref{plaprod}) are one leaf less than the ones in the left, on the contrary,
the sum on leaves of trees in the right of equations (\ref{permcop}) or (\ref{placop}) are one more than that in the left.

In this subsection, we will introduce a new Hopf algebraic structure on $\mathbb{K}[X_{\infty}]$. The operators are different from the Loday-Ronco action, they keep the number of leaves, by the way, the product increases one trivalent vertex, and while the coproduct reduces one trivalent vertex. More remarkable, these definitions can be generalized to graphs with loops in a straightforward way.

\begin{definition}
For $\rho=\langle\rho_1\rho_2\rangle\in \mathbb{K}[X_n] $ and
$\tau=\langle\tau_1\tau_2 \rangle\in \mathbb{K}[X_m] $, then their product
is defined as \be
\rho\star\tau=\langle\rho\tau\rangle+\langle\langle\rho\star\tau_1
\rangle\tau_2\rangle+\langle\rho_1 \langle\rho_2\star\tau \rangle
\rangle ,\ee where $\rho_1, \rho_2, \tau_1, \tau_2\neq \emptyset$. For the righthand side, we add $n$ to the numbers in $\tau$ for
the self-consistency(if the leaves in $\tau$ are labeled by numbers). For any element $\sigma\in \mathbb{K}[X_{\infty}]$,
there are
$$\sigma\star\langle1 \rangle=\langle\sigma1 \rangle
+\langle\sigma_1\langle \sigma_21\rangle \rangle, $$
$$\langle1
\rangle\star\sigma=\langle1\sigma\rangle+\langle\langle
1\sigma_1\rangle \sigma_2\rangle, $$
and
$$\emptyset\star\sigma=\sigma=\sigma\star\emptyset.$$
\end{definition}

\noindent Some examples are listed below.

\begin{align*} &\emptyset\star\langle
1\rangle=\langle 1\rangle=\langle 1\rangle\star\emptyset;\\
&\langle 1\rangle\star\langle 12\rangle=\langle1\langle23 \rangle
\rangle +\langle\langle12 \rangle3 \rangle;\end{align*}
\noindent In terms of diagrams, it is
\begin{align*}
\begin{array}{r}{\epsfxsize
0.6cm\epsffile{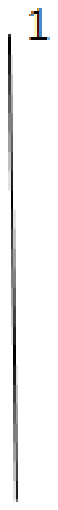}}
\end{array}
\star
\begin{array}{r}
{\epsfxsize 2.2cm\epsffile{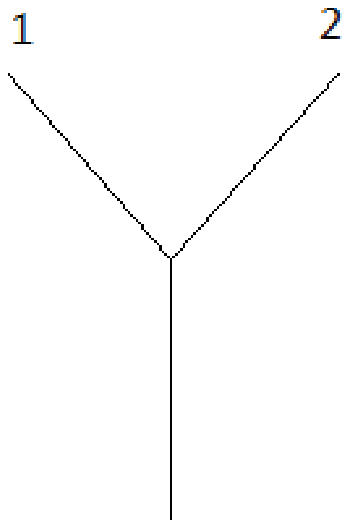}}
\end{array}=
\begin{array}{r}{\epsfxsize
3cm\epsffile{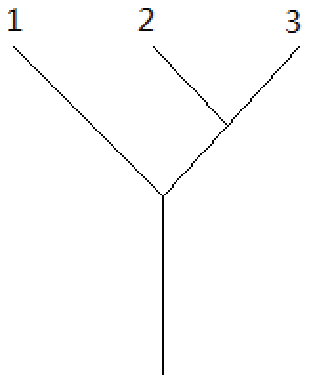}}
\end{array}+\begin{array}{r}{\epsfxsize
3cm\epsffile{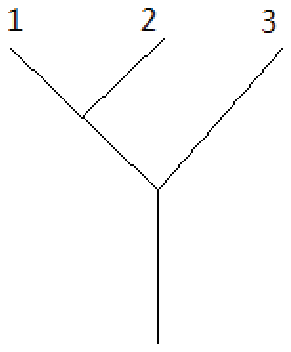}}
\end{array}
\end{align*}
\noindent Similarly
\begin{align*}
\langle12\rangle\star\langle 1\rangle=\langle1\langle23 \rangle \rangle
+\langle\langle12 \rangle3 \rangle;
\end{align*}
\begin{align*}
\begin{array}{r}
{\epsfxsize 2.2cm\epsffile{12.eps}}
\end{array}\star\begin{array}{r}{\epsfxsize
0.6cm\epsffile{1.eps}}
\end{array}=
\begin{array}{r}{\epsfxsize
3cm\epsffile{1-23.eps}}
\end{array}+\begin{array}{r}{\epsfxsize
3cm\epsffile{12-3.eps}}
\end{array}
\end{align*}
\begin{align*}  \langle \langle 12\rangle 3\rangle \star \langle
1\rangle =\langle \langle \langle 12\rangle 3\rangle
4\rangle+\langle \langle 12\rangle\langle 34\rangle\rangle;
\end{align*}
\begin{align*}
\begin{array}{r}{\epsfxsize
3cm\epsffile{12-3.eps}}
\end{array}
\star
\begin{array}{r}
{\epsfxsize 0.7cm\epsffile{1.eps}}
\end{array}=\begin{array}{r}{\epsfxsize
3cm\epsffile{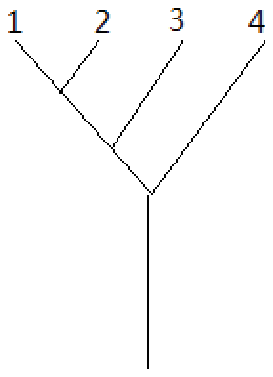}}\end{array}+\begin{array}{r}
{\epsfxsize 3cm\epsffile{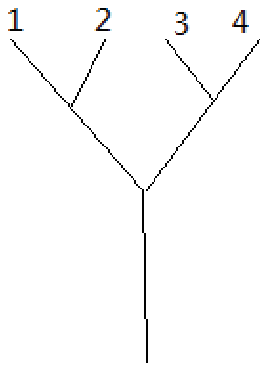}}
\end{array}
\end{align*}
\begin{align*}
&\langle 1\rangle\star\langle \langle 12\rangle 3\rangle=\langle
\langle \langle 12\rangle 3\rangle 4\rangle+\langle \langle 1\langle
23\rangle \rangle 4\rangle +\langle 1\langle \langle 23\rangle
4\rangle;
\end{align*}
\begin{align*}
\begin{array}{r}{\epsfxsize
0.7cm\epsffile{1.eps}}
\end{array}
\star
\begin{array}{r}
{\epsfxsize 3cm\epsffile{12-3.eps}}
\end{array}=\begin{array}{r}{\epsfxsize
3cm\epsffile{12-3-4.eps}}\end{array}+\begin{array}{r}{\epsfxsize
3cm\epsffile{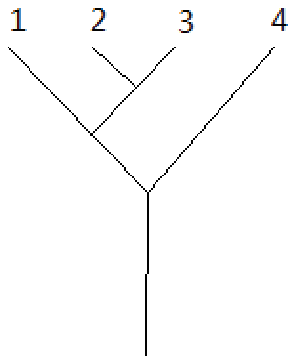}}\end{array}+\begin{array}{r}{\epsfxsize
3cm\epsffile{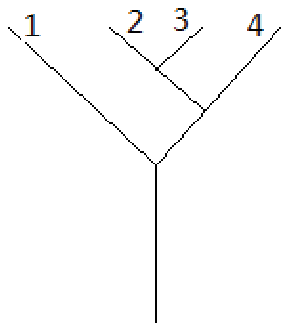}}\end{array}
\end{align*}
\noindent The above two expressions imply that the product $\star$ is non-commutative.
The non-commutate is the essence property of a quantum group.
The product expression for two simplest vertices is
\begin{align*}
\langle 12\rangle\star\langle12\rangle
=\langle \langle 12\rangle\langle 34\rangle\rangle+\langle
\langle \langle 12\rangle 3\rangle 4\rangle +\langle \langle
1\langle 23\rangle \rangle 4\rangle +\langle 1\langle 2\langle
34\rangle \rangle \rangle +\langle 1\langle \langle 23\rangle
4\rangle ;
\end{align*}
\begin{align*}
\begin{array}{r}
{\epsfxsize 2.5cm\epsffile{12.eps}}
\end{array}\star\begin{array}{r}{\epsfxsize
2.5cm\epsffile{12.eps}}
\end{array}=&\begin{array}{r}
{\epsfxsize 3cm\epsffile{12-34.eps}}
\end{array}+\begin{array}{r}{\epsfxsize
3cm\epsffile{12-3-4.eps}}
\end{array}+\begin{array}{r}
{\epsfxsize 3cm\epsffile{1-23-4.eps}}
\end{array}\\&+\begin{array}{r}{\epsfxsize
3cm\epsffile{1--23-4.eps}}
\end{array}+\begin{array}{r}{\epsfxsize
3cm\epsffile{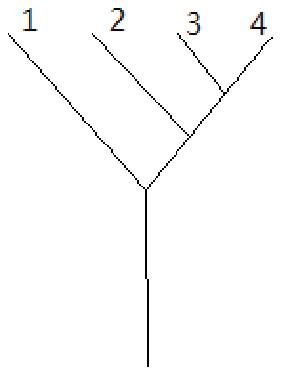}}
\end{array}
\end{align*}

\noindent It is easy to verify that the sums of leaves in two sides of the equations are identical, but the number of vertices is changed.

Even though the product $\star$ on $\mathbb{K}[X_n] $ and the
product $*$ on $\mathbb{K}[Y_n] $ seem to satisfy similar formulas. However, indeed the product $\star$ is very different from $*$. Of course, for the case $\mathbb{K}[Y_{\infty}]$, they are linked by the following identity

\begin{proposition}\label{proprelation}
For arbitrary two planar binary trees $\rho,\sigma\in \mathbb{K}[Y_{\infty}]$, the following identity is preserved
\be\label{prodrelation} \rho\star\sigma=\rho*(1)*\sigma, \ee
where $(1)$ is the planar binary tree with one trivalent vertex in permutation representation.
\end{proposition}

\begin{proof}
The proof is carried out inductively on the number of vertices $k=|\rho|+|\sigma|$.
If $k=0$, we have $\rho=|$ and $\sigma=|$, then the righthand side of equation (\ref{prodrelation}) is equal to the left one, either side is a planar tree with one
trivalent vertex.

If $k=1$, the cases are, either $\rho=|, \sigma=(1)$ or $\rho=(1),
\sigma=|$. It is straightforward to check that, in each case,
any handside of equation (\ref{prodrelation}) is identical
to $\langle \langle 12\rangle 3\rangle+\langle 1\langle 23\rangle
\rangle$.

Suppose the proposition is verified whenever $k<n$. Assume that
for element $\rho=\rho_1\vee\rho_2, \sigma=\sigma_1\vee\sigma_2$ are
arbitrary planar binary trees. If $k=n$, then the right side of the equation
(\ref{prodrelation}) is
\be\ba
\rho*(1)*\sigma=&((\rho*|)\vee|)*\sigma+(\rho_1\vee(\rho_2*(1)))*\sigma\\
=&((\rho\vee|)*\sigma_1)\vee\sigma_2+\rho\vee\sigma+\rho_1\vee(\rho_2*(1)*\sigma)\\
&+((\rho_1\vee(\rho_2*(1)))*\sigma_1)\vee\sigma_2\\
=&(\rho*(1)*\sigma_1)\vee\sigma_2+\rho\vee\sigma+\rho_1\vee(\rho_1*(1)*\sigma)\\
=&\langle \langle \rho\star\sigma_1\rangle \sigma_2\rangle
+\langle\rho\sigma \rangle+\langle \rho_1\langle
\rho_2\star\sigma\rangle \rangle \\
=&\rho\star\sigma. \nonumber\ea\ee
\noindent The first two equalities follow from (\ref{plaprod}), the fourth
equality is governed by the induction. Therefore, we have proved equation (\ref{prodrelation}).
\end{proof}

In the above proposition, $\rho=\emptyset$ or
$\sigma=\emptyset$ is not allowed, for that case there is no definition for $*$.
Now we assume that $\rho,\sigma,\tau$ are two any chosen planar binary trees in
$\mathbb{K}[Y_{\infty}]$. Due to the Proposition (\ref{proprelation}), one has \be\ba
\rho\star(\tau\star\sigma)=&\rho*(1)*(\tau*(1)*\sigma)\\
=&(\rho*(1)*\tau)*(1)*\sigma\\
=&(\rho\star\tau)\star\sigma\nonumber \ea\ee

\noindent If any one of $\rho,\, \sigma$ or $\tau$ is set to be $\emptyset$,
obviously $\rho\star(\tau\star\sigma)=(\rho\star\tau)\star\sigma$. In other
words, the product $\star$ is associative.

We define the unit map $u$ such that $u(1_{\mathbb{K}})=\emptyset$ and
$u(\lambda)=\lambda1_{\mathbb{K}}(\text{for any~}\lambda\in \mathbb{K})$.
By the definition, it is easy to check on $\mathbb{K}[X_{\infty}]$ that
\be \star\circ (u\otimes
\text{id})=\text{id}=\star\circ (\text{id}\otimes u). \ee
\noindent Obviously, $(\mathbb{K}[X_{\infty}],\star,u)$ is an associative algebra.

In order to match the definition of a Hopf algebra, it is required to introduce a coproduct operator on $\mathbb{K}[X_{\infty}]$.

\begin{definition}
For any element $\rho\in \mathbb{K}[X_{n}]$ on a planar binary
tree, assume that $N$ is the set of leaf labels. Then the coproduct of $\rho$ are held with respect to the following equation
\be\ba\label{Bracketcoprod}
\delta(\rho)=\sum_{J\subset  N}\rho_{j}(J)\otimes\rho_{n-j}(N/J), \ea\ee
where $\rho_{|J|}(J)$ is a subgraph with $j=|J|$ leaves labelled with $J$,
and $\rho_{n-j}(N/J)$ is a subgraph with $n-j$ leaves labelled with $N-J$.
These leaves keep on their relationships in $\rho$.
Particularly, $\delta(\emptyset)=\emptyset\otimes\emptyset$,
$\delta(\langle 1\rangle )=\emptyset\otimes\langle 1\rangle +\langle
1\rangle \otimes\emptyset$.
\end{definition}

It is evident that
\begin{align}
\{\rho_J|J\subset N\}=\{\rho_{N/J}|J\subset N\}.
\end{align}
\noindent It is clearly that, $\delta$ and $\Delta$, the definitions of two kinds coproduct are
different. Simple examples are arranged below to illustrate equation (\ref{Bracketcoprod}).

\begin{align*} &\delta\langle 12\rangle
=\emptyset\otimes\langle 12\rangle +\langle 1\rangle \otimes\langle
2\rangle+\langle 2\rangle\otimes\langle 1\rangle +\langle 12\rangle
\otimes\emptyset;\end{align*}
\noindent or in terms of graphs
\begin{align*}
\delta\begin{array}{r}
{\epsfxsize 2.5cm\epsffile{12.eps}}
\end{array}&=\emptyset\otimes\begin{array}{r}
{\epsfxsize 2.5cm\epsffile{12.eps}}
\end{array}+\begin{array}{r}
{\epsfxsize 0.6cm\epsffile{1.eps}}
\end{array}\otimes\begin{array}{r}
{\epsfxsize 0.6cm\epsffile{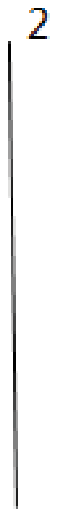}}
\end{array}+\begin{array}{r}
{\epsfxsize 0.6cm\epsffile{2.eps}}
\end{array}\otimes\begin{array}{r}
{\epsfxsize 0.6cm\epsffile{1.eps}}
\end{array}\\ &+\begin{array}{r}
{\epsfxsize 2.5cm\epsffile{12.eps}}
\end{array}\otimes\emptyset
\end{align*}
\noindent For less trivial examples
\begin{align*}
\delta(\langle \langle 12\rangle 3\rangle )=&\emptyset\otimes\langle
\langle 12\rangle 3\rangle+\langle 1\rangle \otimes\langle 23\rangle+\langle 2\rangle \otimes\langle 13\rangle+\langle 3\rangle \otimes\langle 12\rangle
+\langle 12\rangle \otimes\langle 3\rangle\\ &+\langle 13\rangle \otimes\langle 2\rangle+\langle 23\rangle \otimes\langle 1\rangle +\langle \langle
12\rangle 3\rangle \otimes\emptyset;\\
\delta(\langle 1\langle 23\rangle \rangle )=&\emptyset\otimes\langle
1\langle 23\rangle \rangle+\langle 1\rangle \otimes\langle 23\rangle+\langle 2\rangle \otimes\langle 13\rangle+\langle 3\rangle \otimes\langle 12\rangle
+\langle 12\rangle \otimes\langle 3\rangle\\ &+\langle 13\rangle \otimes\langle 2\rangle+\langle 23\rangle \otimes\langle 1\rangle +\langle 1\langle
23\rangle \rangle\otimes\emptyset.
\end{align*}
\noindent and equivalently expressed through graphs
\begin{align*}
\delta\begin{array}{r}
{\epsfxsize 3cm\epsffile{12-3.eps}}
\end{array}&=\emptyset\otimes\begin{array}{r}
{\epsfxsize 3cm\epsffile{12-3.eps}}
\end{array}+\begin{array}{r}
{\epsfxsize 0.6cm\epsffile{1.eps}}
\end{array}\otimes\begin{array}{r}
{\epsfxsize 2.5cm\epsffile{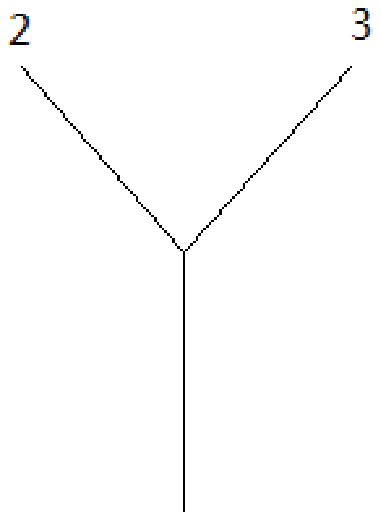}}
\end{array}+\begin{array}{r}
{\epsfxsize 0.6cm\epsffile{2.eps}}
\end{array}\otimes\begin{array}{r}
{\epsfxsize 2.5cm\epsffile{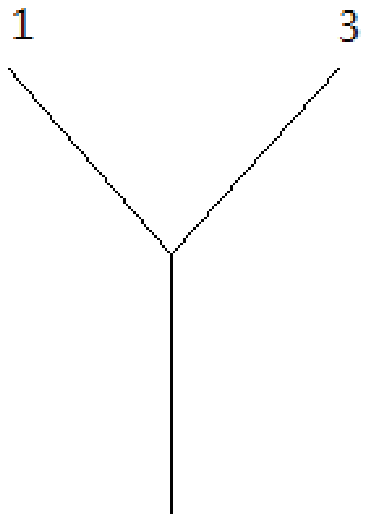}}
\end{array}\\ +&\begin{array}{r}
{\epsfxsize 0.6cm\epsffile{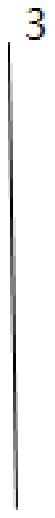}}
\end{array}\otimes\begin{array}{r}
{\epsfxsize 2.5cm\epsffile{12.eps}}
\end{array}+\begin{array}{r}
{\epsfxsize 2.5cm\epsffile{12.eps}}
\end{array}\otimes\begin{array}{r}
{\epsfxsize 0.6cm\epsffile{3.eps}}
\end{array}+\begin{array}{r}
{\epsfxsize 2.5cm\epsffile{13.eps}}
\end{array}\otimes\begin{array}{r}
{\epsfxsize 0.6cm\epsffile{2.eps}}
\end{array}\\ +&\begin{array}{r}
{\epsfxsize 2.5cm\epsffile{23.eps}}
\end{array}\otimes\begin{array}{r}
{\epsfxsize 0.6cm\epsffile{1.eps}}
\end{array}+\begin{array}{r}
{\epsfxsize 3cm\epsffile{12-3.eps}}
\end{array}\otimes\emptyset\ \ \
\end{align*}

It is clear that the numbers of leaves on the left and right sides are intact for each equation, while the numbers of vertices are changed. Using the definition of the coproduct $\delta$, for any given tree $\rho$, one has
\begin{align*}
(\text{id}\otimes\delta)\circ\delta\rho(K)=&(\text{id}\otimes\delta)\circ\sum_{J\subset K}\rho_{|J|}(J)\otimes\rho_{k-|J|}(K/J)\\
=&\sum_{J\subset K}\rho_{|J|}(J)\otimes(\sum_{I\subset K/J}\rho_{|I|}(I)\otimes\rho_{k-|J|-|I|}(K/(J\cup I)))\\
=&\sum_{J\subset K}\sum_{I\subset K/J}\rho_{|J|}(J)\otimes\rho_{|I|}(I)\otimes\rho_{k-|J|-|I|}(K/(J\cup I))
\end{align*}
and
\begin{align*}
(\delta\otimes\text{id})\circ\delta\rho(K)=&(\delta\otimes\text{id})\circ\sum_{J\subset K}\rho_{|J|}(J)\otimes\rho_{k-|J|}(K/J)\\
=&\sum_{J\subset K}\sum_{I\subset J}\rho_{|I|}(I)\otimes\rho_{|J|-|I|}(J/I)\otimes\rho_{k-|J|}(K/J)
\end{align*}
\noindent Hence, for any chosen $\rho\in \mathbb{K}[X_{\infty}]$, the identity
\be\label{eq1}
(\text{id}\otimes\delta)\circ\delta(\rho)=(\delta\otimes\text{id})\circ\delta(\rho).
\ee
is held. Now, we can introduce a counit map $\varepsilon$ in $\mathbb{K}[X_{\infty}]$. For any element $\sigma\in \mathbb{K}[X_{\infty}]$, we set
\be\label{eq3}\varepsilon(\sigma)=
\begin{cases}1_{\mathbb{K}} &\text{if}\ \sigma=\emptyset\\
0  &\text{otherwise}.\end{cases}
\ee
By equation (\ref{eq3}), we can easily get $(\varepsilon\otimes\text{id})\circ
\delta=\text{id}=(\text{id}\otimes\varepsilon)\circ
\delta$. Therefore, $(\mathbb{K}[X_{\infty}],\delta,\varepsilon)$ is a counital coalgebra.

Now $\mathbb{K}[X_{\infty}]$ is an algebra as well as a coalgebra. For our target, the algebraic operators and the coalgebraic operators should form a bialgebra.

\begin{proposition}
$(\mathbb{K}[X_{\infty}],\star, u,
\delta,\varepsilon)$ is a bialgebra.
\end{proposition}
\begin{proof}
We just need to prove $\star$ is a morphism of coalgebra
\be\label{eq4}
\delta\circ\star=(\star\otimes\star)\circ\tau_{2,3}\circ(\delta\otimes\delta)
\ee
and
\be\label{eq5}
\varepsilon\otimes\varepsilon=\star\circ(\varepsilon\otimes\varepsilon)=\varepsilon\circ\star.
\ee
By using of the definition of $\varepsilon$, equation (\ref{eq5}) is obvious. Now we pay attention on equation (\ref{eq4}). Let $K,L$ be sets of labels and $|K|=k,|L|=l$. For any two picked out trees $\rho(K)=\langle\rho_a\rho_b\rangle,\sigma(L)=\langle\sigma_a\sigma_b\rangle$ in $ X_{\infty}$, the left side of equation (\ref{eq4}) is
\begin{align*}
\delta(\rho\star\sigma)=\delta\langle\rho\sigma\rangle+\delta\langle\rho_a(\rho_b\star\sigma)\rangle+\delta\langle(\rho\star\sigma_a)\sigma_b\rangle.
\end{align*}
The right side of equation (\ref{eq4}) is equal to
\begin{align*}
&\star\otimes\star\circ\tau_{23}\circ\delta\otimes\delta(\rho\otimes\sigma)\\ =&\star\otimes\star\circ\tau_{23}\circ(\sum_{I\subset K}\rho_{|I|}\otimes\rho_{k-|I|})\otimes(\sum_{J\subset L}\sigma_{|J|}\otimes\sigma_{l-|J|})\\ =&\sum_{I\subset K}\sum_{J\subset L}\star\otimes\star\circ(\rho_{|I|}\otimes\sigma_{|J|})\otimes(\rho_{k-|I|}\otimes\sigma_{l-|J|})\\ =&\sum_{I\subset K}\sum_{J\subset L}[\langle\rho_{|I|}\sigma_{|J|}\rangle+\langle\rho_{a|I_{1}|}\langle\rho_{b|I_{2}|}\sigma_{|J|}\rangle\rangle+\langle\langle\rho_{|I|}\sigma_{a|J_{1}|}\rangle\sigma_{b|J_{2}|}]\otimes\\
&[\langle\rho_{k-|I|}\sigma_{l-|J|}\rangle+\langle\rho_{a|I_{3}|}\langle\rho_{b|I_{4}|}\sigma_{|J|}\rangle\rangle+\langle\langle\rho_{|I|}\sigma_{a|J_{3}|}\rangle\sigma_{b|J_{4}|}]\\
=&\delta[\langle\rho\sigma\rangle+\langle\rho_a(\rho_b\star\sigma)\rangle+\langle(\rho\star\sigma_a)\sigma_b\rangle].
\end{align*}
where $I_{1}\cup I_{2}=I$; $I_{3}\cup I_{4}=K/I$; $J_{1}\cup J_{2}=J$; $J_{3}\cup J_{4}=L/J$. That equation (\ref{eq4}) is held.

\end{proof}

In fact, $\mathbb{K}[X_{\infty}]$ is a graded bialgebra. Using the definition, we know $\mathbb{K}[X_{\infty}]=\sum_{n\in
N_0}L^n$, where $L^0=\{\emptyset\}, L^1=\{\langle 1\rangle \}, L^2=X_1,\cdots,
L^n=X_{n-1}\cdots$. Since $L^0$ is one dimensional vector space generated by $\emptyset$, then $\mathbb{K}[X_{\infty}]$ is a connected bialgebra. In term of another word, $(\mathbb{K}[X_{\infty}],\star, u, \delta,\varepsilon)$ is a Hopf algebra.

\subsection{Hopf algebra for graphs with loops}

In this subsection, we derive the Hopf algebra on graphs with loops. With an action of the contraction operators, a planar binary tree generates planar binary graph with loops.
The new Hopf algebra introduced at moment on planar binary trees can be generalized
to graphs with loops.

Let $\bar X_n$ mark the set of graphs obtained from $X_n$ by making all possible contractions $(_{i}\leftrightarrow_{i+1})$ on near leaves. It is obvious that $\bar X_0=X_0$,
$\bar X_1=X_1$, $X_n\subset \bar X_n(n>1)$. And mark again $\bar X^g_{n-2g}$ the set
of trees contracted all possible $g(g\leq \lfloor{n\over 2}\rfloor)$ pairs with the nearest
leaves of trees in $\bar X_n$. From the section 4, one has $\bar X^g_{n-2g}=Y_{n-2g-1}^g=\mathscr{G}_{n-2g-1}^g$. With the notations introduced, $\mathbb{K}[\bar X_n]$
is the vector space generated by $\bar X_n$ over $\mathbb{K}$. The disjoint union of $\bar X_n(n\geq 0)$ is $\bar X_{\infty}=\cup_{n\in N_0}\bar X_n$, and its vector space over $\mathbb{K}$ is $\mathbb{K}[\bar X_{\infty}]$. By the definition, $\mathbb{K}[ X_{\infty}]\subset\mathbb{K}[\bar X_{\infty}]$. Now we extend the structure $(\mathbb{K}[X_{\infty}],\star, u, \delta,\varepsilon)$ to $\mathbb{K}[\bar X_{\infty}]$.

Now we define the product and the unit map on $\mathbb{K}[\bar X_{\infty}]$.
The operator $\star$ on $\mathbb{K}[X_{\infty}]$ does not effect the adjacent
relationships between leaves on planar binary trees. In other words, the contraction
operator keeps the adjacent positions invariant. Hence the contraction is commutative
with the product $\star$. Therefore, we can define the product and the unit map on $\mathbb{K}[\bar X_{\infty}]$ as following.
\begin{align}\label{eq6}
\star|_{\mathbb{K}[\bar X_{\infty}]}=\sum_{i}(_{i}\leftrightarrow_{i+1})\circ \star|_{\mathbb{K}[X_{\infty}]};\nonumber \\
u|_{\mathbb{K}[\bar X_{\infty}]}=\sum_{i}(_{i}\leftrightarrow_{i+1})\circ u|_{\mathbb{K}[X_{\infty}]}.
\end{align}
Evidently, $(\mathbb{K}[\bar X_{\infty}],\star, u)$ is a unital algebra as expected.

Following the process to put the product operator on $\mathbb{K}[\bar X_{\infty}]$, we can  introduce the coproduct operator and the counit map by their commutation with the contraction operator.
\begin{align}\label{eq7}
\delta|_{\mathbb{K}[\bar X_{\infty}]}=\sum_{i}(_{i}\leftrightarrow_{i+1})\circ \delta|_{\mathbb{K}[X_{\infty}]};\nonumber \\
\varepsilon|_{\mathbb{K}[\bar X_{\infty}]}=\sum_{i}(_{i}\leftrightarrow_{i+1})\circ \varepsilon|_{\mathbb{K}[X_{\infty}]}.
\end{align}
It gives $\delta(\emptyset)=\emptyset\otimes\emptyset$, $\delta\langle1\rangle=\emptyset\otimes\langle1\rangle+\langle1\rangle\otimes\emptyset$.

Given any particular planar binary graph with loops, we can realize its coproduct
following a  simple procedure. As an example, now we carry out $\delta\langle\overline{12}\rangle$. The first step, we deal with the coproduct on
the tree $\langle12\rangle$,
$$\delta\langle12\rangle=\emptyset\otimes\langle12\rangle+\langle1\rangle\otimes\langle2\rangle+\langle2\rangle\otimes\langle1\rangle+\langle12\rangle\otimes\emptyset.$$
Then, we perform the same contraction as on the primary tree, that links leaves $1,2$
by a line. In another word, we have
\begin{align}
\delta[\langle\overline{12}\rangle]=\delta[\langle\overline{11}\rangle]=&
\emptyset\otimes\langle\overline{11}\rangle+\overline{\langle1\rangle\otimes\langle1\rangle}+\overline{\langle1\rangle\otimes\langle1\rangle}+\langle\overline{11}\rangle\otimes\emptyset,\nonumber
\end{align}

\noindent where two leaves with the same label under a bar imply that they are
joined by a line to make a loop. Of course, it is trivial that $\overline{\langle1\rangle\otimes\langle1\rangle}=\langle1\rangle\otimes\emptyset$. The graphic
representation for the coproduct on the simplest diagram with one loop is

\begin{align*}
\delta\begin{array}{r}
{\epsfxsize 2.5cm\epsffile{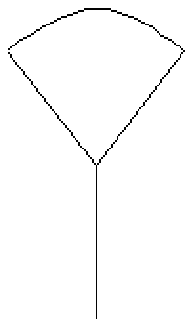}}
\end{array}=\emptyset\otimes\begin{array}{r}
{\epsfxsize 2.5cm\epsffile{0.eps}}
\end{array}+\begin{array}{r}
{\epsfxsize 0.7cm\epsffile{1.eps}}
\end{array}\otimes\emptyset+\begin{array}{r}
{\epsfxsize 0.7cm\epsffile{1.eps}}
\end{array}\otimes\emptyset+\begin{array}{r}
{\epsfxsize 2.5cm\epsffile{0.eps}}
\end{array}\otimes\emptyset
\end{align*}

\noindent Noe we deal with $\delta\langle1\langle\overline{23}\rangle\rangle$, a less
simpler example. At first, we expand the coproduct of the tree with no contractions.

\begin{align}
\delta(\langle 1\langle 23\rangle \rangle )=&\emptyset\otimes\langle
1\langle 23\rangle \rangle+\langle 1\rangle \otimes\langle 23\rangle+\langle 2\rangle \otimes\langle 13\rangle+\langle 3\rangle \otimes\langle 12\rangle
+\langle 12\rangle \otimes\langle 3\rangle\\ &+\langle 13\rangle \otimes\langle 2\rangle+\langle 23\rangle \otimes\langle 1\rangle +\langle 1\langle
23\rangle \rangle\otimes\emptyset.\nonumber
\end{align}

\noindent and,

\begin{align}
\delta\langle1\langle\overline{23}\rangle\rangle=&\emptyset\otimes\langle1\langle\overline{22}\rangle\rangle
+\langle1\rangle\otimes\langle\overline{22}\rangle+2\overline{\langle2\rangle\otimes\langle12\rangle}+2\langle1\overline{2\rangle\otimes\langle2\rangle}
+\langle\overline{22}\rangle\otimes\langle1\rangle \nonumber \\ &+\langle1\langle\overline{22}\rangle\rangle\otimes\emptyset,
\end{align}

\begin{align*}
&\delta\begin{array}{r}
{\epsfxsize 2.5cm\epsffile{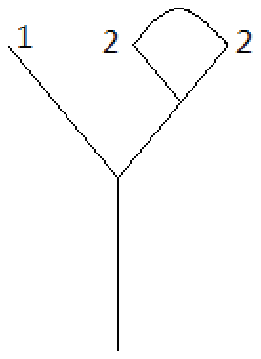}}
\end{array}=\emptyset\otimes\begin{array}{r}
{\epsfxsize 2.5cm\epsffile{1-0.eps}}
\end{array}+\begin{array}{r}
{\epsfxsize 0.7cm\epsffile{1.eps}}
\end{array}\otimes\begin{array}{r}
{\epsfxsize 2.5cm\epsffile{0.eps}}
\end{array}+4\begin{array}{r}
{\epsfxsize 2.5cm\epsffile{12.eps}}\nonumber
\end{array}
\end{align*}
\begin{align*}
\otimes\emptyset+\begin{array}{r}
{\epsfxsize 2.5cm\epsffile{0.eps}}
\end{array}\otimes\begin{array}{r}
{\epsfxsize 0.7cm\epsffile{1.eps}}
\end{array}+\begin{array}{r}
{\epsfxsize 2.5cm\epsffile{1-0.eps}}
\end{array}\otimes\emptyset
\end{align*}

\noindent It is easy to verify that $(\mathbb{K}[\bar X_{\infty}],\delta,\varepsilon)$ is a counital algebra. Next,  we will prove that it is a Hopf algebra.
\begin{proposition}\label{hopfproof}
$(\mathbb{K}[\bar X_{\infty}],\star, u, \delta, \varepsilon)$ is a Hopf algebra.
\end{proposition}
\begin{proof}
At first, we prove $(\mathbb{K}[\bar X_{\infty}],\star, u, \delta, \varepsilon)$ forms a bialgebra.
As we have announced $\mathbb{K}[\bar X_{\infty}]$ is graded, as well as connected,
in the other word, $(\mathbb{K}[\bar X_{\infty}],\star, u, \delta, \varepsilon)$ is a Hopf
algebra. Hence, it is sufficient to prove that equation (\ref{eq4}) is made sense on $\mathbb{K}[\bar X_{\infty}]$.

For the leaves are unlabelled, the operator $(_{i}\leftrightarrow_{i+1})$ is commutative with the permutations $\tau_{23}$. For any given two trees $\bar\rho,\bar\sigma\in \mathbb{K}[\bar X_{\infty}]$, and the corresponding primary trees $\rho=(\rho_1\rho_2),\sigma=(\sigma_1\sigma_2)\in \mathbb{K}[X_{\infty}]$.
\begin{align}
\delta\langle\bar\rho\star\bar\sigma\rangle
=&\delta\circ(\sum_{\rho,\delta}(_{i}\leftrightarrow_{i+1}))\circ(\rho\star\delta)\nonumber\\
=&\sum_{\rho,\delta}(_{i}\leftrightarrow_{i+1})\circ\delta(\rho\star\delta)\nonumber
\end{align}
On the other hand
\begin{align*}
\star\otimes\star\circ\tau_{23}\circ\delta\otimes\delta(\bar\rho\otimes\bar\sigma)=&
\star\otimes\star\circ\tau_{23}\circ\delta\otimes\delta(\sum_{\rho,\delta}(_{i}\leftrightarrow_{i+1}))\circ(\rho\star\delta)\\
=&\sum_{\rho,\delta}(_{i}\leftrightarrow_{i+1})\circ
\star\otimes\star\circ\tau_{23}\circ\delta\otimes\delta(\rho\otimes\sigma)
\end{align*}
Using the equality (\ref{eq4}), one gets
\begin{align*}
\delta\langle\bar\rho\star\bar\sigma\rangle=
\star\otimes\star\circ\tau_{23}\circ\delta\otimes\delta(\bar\rho\otimes\bar\sigma).
\end{align*}
\end{proof}
\noindent Immediately, using the structure on on $\mathbb{K}[\bar X_{\infty}]$, we can dress the
Hopf algebra structure on the topological recursion by using the weighted map $\phi$.
The Hopf algebra on diagrammatic representations of the correlation functions is one of the results of the Hopf algebra structure on $\mathbb{K}[\bar X_{\infty}]$.

With the help of Hopf algebra on $\mathbb{K}[\bar X_{\infty}]$, immediately, as a corollary, we achieve a more concise result in contrast to that in  \cite{Est}. The formula can be expressed
as the following theorem.

\begin{corollary}
The diagrammatic representation of the correlation function $W_k^g(p,p_1,\cdots,p_{k-1})$ is given by summing over all permutations $k-1$ labelled leaves for graphs in $(X_{k-1})^g$,
\begin{align}\label{equ1}
\phi^{-1} (W_k^g(p,p_1,\cdots,p_{k-1}))=\sum_{\scriptstyle t^g\in
(X_{k-1})^g\atop \scriptstyle \text{perm. $\{p_1,\cdots,
p_{k-1}\}$}}t^g,
\end{align}
in which $(X_{k-1})^g$ is the set of graphs obtained by the product $\star$ of $(2g+k-1)$ $\langle1\rangle$ with all possible $g$ pairs of the nearest leaves identifications,
\begin{align}\label{equ2}
X_{k-1}^g=\{\underbrace{\langle1\rangle\star\langle1\rangle\star\cdots
\star\langle1\rangle}_{\mathrm{2g+k-1}}
|\text{all possible contractions of $g$ pairs of the nearest leaves}\}.
\end{align}
\end{corollary}

\begin{proof}
It is easy to see that $X^g_{k-1}$ is identified with $(Y^{2g+k-2})^g$ in  \cite{Est}.
Using the result in  \cite{Est}, we know that equation (\ref{equ1}) is sured.
For any element $\rho\in\mathbb{K}[\bar X_{\infty}]$, it is straightforward
$$\rho\star\langle1\rangle=\rho*\langle12\rangle.$$
Hence, equation (\ref{equ2}) is followed directly the Theorem 2 in  \cite{Est}. It should be noticed that the number of $\langle1\rangle$ in equation (\ref{equ2}) is one more than the number of $(1)$ in Theorem 2 in  \cite{Est}.
\end{proof}

\subsection{Antipode}

For a graded and connected Hopf algebra, there is a canonical antipode $S$ governed
by the following expression:

\begin{align*}
\star\circ(S\otimes Id)\circ\delta=u\circ\varepsilon.
\end{align*}
For any element $\rho\in \mathbb{K}[X_{\infty}]$, we have
\begin{align}
S(\rho)=-\rho-\sum S(\rho_1)\star\rho_2,
\end{align}
where in the Sweedler notation
\begin{align*}
\delta\rho=\sum\rho_1\otimes\rho_2
\end{align*}
\noindent is assumed. It is easy to see that
\begin{align*}
&S(\emptyset)=\emptyset;\\
&S(\langle 1\rangle)=-\langle1\rangle;\\
&S(\langle 12\rangle)=\langle 21\rangle;\\
&S(\langle\langle 12\rangle3\rangle)=\langle1\langle 23\rangle\rangle-\langle\langle 32\rangle1\rangle-\langle3\langle 21\rangle\rangle.
\end{align*}
With the help of the identity
\begin{align*}
S\circ \star=\star\circ \tau \circ (S\otimes S),
\end{align*}
here the $\tau$ is a permutation. Therefore, we have
\begin{align}
S(\langle 1\rangle\star\langle2\rangle\star\cdots\star\langle n\rangle)=(-1)^{n}\langle n\rangle\star\langle n-1\rangle\star\cdots\star\langle1\rangle.
\end{align}
And then the weighted map $\phi^{-1*}S$ induced by the antipode on the vector space of
the correlation functions implies
\begin{align}
(\phi^{-1*}S)W^0_{n+1}(p,p_1,\cdots,p_n)=(-1)^nW^0_{n+1}(p,p_1,\cdots,p_n).
\end{align}
As the action $(_{i}\leftrightarrow_{i+1})$ is commutative with all of the the product $\star$,
the coproduct $\delta$, the unit map $u$ as well as the counit map $\varepsilon$. Hence, the contraction operator is also commutative with the antipode $S$.
As an example
\begin{align*}
S((_{1}\leftrightarrow_{2})\langle{12}\rangle)=(_{1}\leftrightarrow_{2})S(\langle{12}\rangle)=(_{1}\leftrightarrow_{2})\langle{21}\rangle.
\end{align*}
\noindent Therefore, we have
\begin{align*}
(\phi^{-1*}S)W^g_{n+1}(p,p_1,\cdots,p_n)=(-1)^{n+2g}W^g_{n+1}(p,p_1,\cdots,p_n).
\end{align*}

\subsection{Coproduct and topological recursion }

At first sight, it seems that the expression of coproduct $\delta$ defined on the tagged
graphs by Equations (\ref{Bracketcoprod}), (\ref{eq7}) and the topological recursion are
very different. The main discrepancy between them is that the numbers of their diagram representation are different. It looks like that this unmatched point cannot be eliminated.
However, after a careful inspection, In fact, they share numerous common properties.
The unmatched point will be removed by considering the symmetry of the the tagged
graphs. If we tie with coefficients to terms of coproduct for the tagged
graphs, the topological recursion can be very naturally reconstructed.

Given any graph $G_1\in \mathscr{G}_{k}^{g}(p,p_1,\dots,p_k) $, If we introduce $\delta'G_1=\delta G_1-\emptyset\otimes G_1-G_1\otimes\emptyset$. In other words,
\begin{align}{\label{changeproduct}}
\delta'G_1(p,K)=\sum_{I\subset K}\rho_{1I}\otimes\bar\rho_{1K/I},
\end{align}
where $K=\{p_1, p_2,\cdots, p_k\}$. For simplicity, we collect the terms with equal leaves in Equation (\ref{changeproduct}) and set $T_i\otimes\bar T_{k-i}=\sum_{I\subset K, |I|=i}\rho_I\otimes\bar \rho_{K/I}$ formally, then one has the following proposition.

\begin{proposition}{\label{protop1}}
If
\begin{align}
\sum_{G\in \mathscr{G}_{k}^{0}}\delta'G(p,K)=\sum_{i=1}^{k-1}T_{i}\otimes\bar T_{k-i},
\end{align}
then the following equation holds
\begin{align}
\phi(\sum_{i=1}^{i=k-1}a_iT_{i}\otimes\bar T_{k-i})=\mathop{\text{Res}}\limits_{q\rightarrow \bf {a}}K_2(p;q)\sum_{I\subset K}W_{|I|+1}(q,I)W_{k-i+1}(\bar q, K/I)=W_{k+1}(p,K),
\end{align}
where the coefficients $a_i={C_iC_{k-i-1}\over C_{k-1}\tbinom{k}{i}}$, and $C_i$ is the Catalan number.
\end{proposition}
\begin{proof}
As the labels in $K$ are symmetrical, the graphs in $T_{i}\bar T_{k-i}$ are also symmetrical with respect to labels in $K$. By the definition of the coproduct $\delta$, there are much more the tagged graphs in $T_{i}\bar T_{k-i}$ than the diagrammatic representations of the topological recursion $W_{|I|+1}(q,I)W_{k-i+1}(\bar q, K/I)$. However, due to the symmetry, if we define $a_i$ as the ratio of graphs in diagrammatic representations of the topological recursion $W_{|I|+1}(q,I)W_{k-i+1}(\bar q, K/I)$ to the tagged graphs in $T_{i}\bar T_{k-i}$. With the help of results in \cite{Be}, one has
\begin{align*}
a_i={\tbinom{k}{i}i!(k-i)!C_{i-1}C_{n-i-1}\over k!C_{k-1}\tbinom{k}{i}}={C_iC_{k-i-1}\over C_{k-1}\tbinom{k}{i}}.
\end{align*}
\end{proof}

Of course, the results can be extended to the correlation function. For arbitrary correlation function of order to genus $g$, we have similar results.
\begin{proposition}{\label{protop2}}
If
\begin{align}
\sum_{G\in \mathscr{G}_{k}^{g}}\delta'G(p,K)=\sum_{m=0}^g\sum_{i=1}^{i=k-1}T^m_{i}\otimes\bar T^{m-g}_{k-i}+T_{k+1}^{g-1},
\end{align}
then, the identity holds
\begin{align}
&\phi(\sum_{m=0}^g\sum_{i=1}^{i=k-1}a_i^mT^m_{i}\otimes\bar T^{m-g}_{k-i}+b_{k+1}^{g-1}T_{k+1}^{g-1})\nonumber\\=&\mathop{\text{Res}}\limits_{q\rightarrow \bf {a}}K_2(p;q)\sum_{m=0}^g\sum_{I\subset K}W^m_{|I|+1}(q,I)W^{m-g}_{k-|I|+1}(\bar q, K/I)+W_{k+2}^{g-1}(q,\bar q, K)\nonumber\\=&W^g_{k+1}(p,K),
\end{align}
where the coefficients $a^m_i={s_ms_{g-m}\tbinom{{3\over 2}(m-1)+i}{i}\tbinom{{3\over 2}(g-m-1)+k-i}{k-i}\over s_g\tbinom{k}{i}\tbinom{g}{m}\tbinom{{3\over 2}(g-1)+k}{k}}$, while $b_{k+1}^{g-1}={g\over 2^k(4^g-2^g)}$, and $C_i$ is the Catalan number. The parameters $s_m$ for $m\geq 1$ obey \cite{Be}:
\begin{align*}
s_1=1,\ \ \ \ s_m=2(3m-4)s_{m-1}+\sum_{n=1}^{m-1}s_ns_{m-n}.
\end{align*}
\end{proposition}
\begin{proof}
Following the same process as that in the proposition (\ref{protop1}), we just list out
the ratio of graphs on diagrammatic representation of the topological recursion to the
tagged graphs with the coproduct.
\begin{align*}
a^m_i=&{\tbinom{k}{i}s_ms_{g-m}i!(k-i)!4^k\tbinom{{3\over 2}(m-1)+i}{i}\tbinom{{3\over 2}(g-m-1)+k-i}{k-i}\over s_gk!4^k\tbinom{k}{i}\tbinom{g}{m}\tbinom{{3\over 2}(g-1)+k}{k}}\\=&{s_ms_{g-m}\tbinom{{3\over 2}(m-1)+i}{i}\tbinom{{3\over 2}(g-m-1)+k-i}{k-i}\over s_g\tbinom{k}{i}\tbinom{g}{m}\tbinom{{3\over 2}(g-1)+k}{k}},
\end{align*}
in which, the numerator counts the graphs on diagrammatic representation of the topological recursion $\sum_{I\subset K}W^m_{|I|+1}(q,I)$ $\times W^{m-g}_{k-|I|+1}(\bar q, K/I)$, the denominator counts the tagged graphs in the coproduct of $T^g_k$.
Meanwhile
\begin{align*}
b_{k+1}^{g-1}=&{s_gk!4^k\tbinom{{3\over 2}(g-1)+k}{k}\over 2^k(4^g-2^g)s_gk!4^k\tbinom{{3\over 2}(g-1)+k}{k}}\\=&{1\over 2^k(4^g-2^g)}.
\end{align*}
By the actual meaning of the ratio, ${1\over a^m_i},  {1\over b_{k+1}^{g-1}}$ are integers. That is obvious for $b_{k+1}^{g-1}$.
\end{proof}

With these formula, we set up the concrete relationship between the Hopf algebraic
structure defined on the tagged graphs and the topological recursion. Of course,
the numbers of the tagged graphs are much more than the ones of the diagram represents
for the topological recursion. This implies that the diagram represents for the topological
recursion is very special case of the tagged graphs with a coproduct structure.
The tagged graphs could be involved with more mathematical information and structures.

\section{Conclusions and discussions}

In this article, we propose a new Hopf algebraic structure on the tagged graphs with or without loops, the topological recursion on a spectral curve is reconstructed by considering the symmetry
of the the tagged graphs. In certain sense, the topological recursion on an arbitrary algebraic
curve is a special case of the new version Hopf algebraic structure.

The formal hermitian matrix integrals are important models in mathematical physics, it is the
zero-dimensional quantum field theory. Hence, it is naturally to consider that the Hopf
algebraic structures on the Feynman graphs of quantum field theory obtained by
A. Connes and D. Kreimer. The Connes-Kreimer Hopf algebra originates on trees
diagrams as well, maybe it has intrinsic relationship with the Hopf algebra on topological recursion.

However, compared with the Connes-Kreimer Hopf algebra on the Feynman diagrams of quantum field theory, the operators here have completely different meaning for the
properties of the graphs. As we explore in this paper, the new Hopf algebraic structure on the tagged graphs with loops is very different from the Loday-Ronco Hopf algebra.
In a sense, the new Hopf algebra can be regarded as a dual version of the
Loday-Ronco Hopf algebra. While it is known that, the Loday-Ronco Hopf algebra
is isomorphic to the non-commutative Connes-Kreimer Hopf algebra \cite{AS}.
Like a quantum group structure, the non-commutate is the essence property for the
Hopf algebraic structure introduced in this paper. Therefore It is interesting to
investigate the relationships between the new Hopf algebra
and the non-commutative Connes-Kreimer Hopf algebra.

It is worth exploring the reasons of a Hopf algebra existence in the context of the topological recursions as well as the Feynman diagram perturbative expansion in quantum field theory, this is helpful to understand the phenomena of topological recursions. Whether the properties on topological recursions can be established in perturbative quantum field theories? There are a lot of problems remained to be clarified yet.

It needs to mention that, in this paper we limit to consider the topological recursion which is defined on algebraic curves with simple branch points. In this case, the recursion relations are represented by diagrams with trivalent vertices \cite{EO}. For more general topological recursions, in which there are vertices beyond three-valent vertices \cite{BHLMR}, maybe there is a Hopf algebraic structure, too. We hope to consider that problem in our future work.

\section*{Acknowledgments}
The financial supports from the National Natural Science Foundation
of China (NSFC, Grants 11375258 and 11775299) are gratefully acknowledged
from one of the author (Ding). Meng was supported by NSFC Tianyuan Special fund (Grant 11626222) and he would like to thank J. N. Esteves for the explanation of his work  \cite{Est}.

\end{document}